\pdfoutput=1
\documentclass{article}
\usepackage[utf8]{inputenc}
\usepackage{amsmath}
\usepackage{amssymb}
\usepackage{graphicx}
\usepackage{amsthm}
\newtheorem{theorem}{Theorem}

\newtheorem{lemma}{Lemma}
\newtheorem{example}{Example}

\usepackage{algpseudocode}
\usepackage{algorithmicx}
\usepackage{algorithm}

\newcommand{\beq}{\begin{equation}}
\newcommand{\eeq}{\end{equation}}
\newcommand{\barr}{\left[\begin{array}}
\newcommand{\earr}{\end{array}\right]}

\newcommand{\rank}{\mbox{rank}\,}

\newcommand{\bpf}{\begin{proof}}
\newcommand{\epf}{\end{proof}}
\newcommand{\ord}{\mbox{ord}\,}

\newcommand{\ftwo}{\ensuremath{\mathbb{F}_{2}}}

\newcommand{\ff}{\ensuremath{\mathbb{F}}}
\newcommand{\al}{\alpha}

\begin{document}
\title{Word Linear Complexity of sequences and\\ Local Inversion of maps over finite fields}
\author{Virendra Sule\\Professor (Retired)\\Department of Electrical Engineering\\
Indian Institute of Technology Bombay\\Mumbai 400076, India\\vrs@ee.iitb.ac.in}
%\date{August 17, 2023}
%started August 17, 2023
\maketitle

\begin{abstract}
This paper develops the notion of \emph{Word Linear Complexity} ($WLC$) of vector valued sequences over finite fields $\ff$ as an extension of Linear Complexity ($LC$) of sequences and their ensembles. This notion of complexity extends the concept of the minimal polynomial of an ensemble (vector valued) sequence to that of a matrix minimal polynomial and shows that the matrix minimal polynomial can be used with iteratively generated vector valued sequences by maps $F:\ff^n\rightarrow\ff^n$ at a given $y$ in $\ff^n$ for solving the unique local inverse $x$ of the equation $y=F(x)$ when the sequence is periodic. The idea of solving a local inverse of a map in finite fields when the iterative sequence is periodic and its application to various problems of Cryptanalysis is developed in previous papers \cite{sule322, sule521, sule722,suleCAM22} using the well known notion of $LC$ of sequences. $LC$ is the degree of the associated minimal polynomial of the sequence. The generalization of $LC$ to $WLC$ considers vector valued (or word oriented) sequences such that the word oriented recurrence relation is obtained by matrix vector multiplication instead of scalar multiplication as considered in the definition of $LC$. Hence the associated minimal polynomial is matrix valued whose degree is called $WLC$. A condition is derived when a nontrivial matrix polynomial associated with the word oriented recurrence relation exists when the sequence is periodic. It is shown that when the matrix minimal polynomial exists $n(WLC)=LC$. Finally it is shown that the local inversion problem is solved using the matrix minimal polynomial when such a polynomail exists hence leads to a word oriented approach to local inversion.
\end{abstract}

\emph{Subject Classification}: cs.CC, cs.CR, math.NT
\section{Introduction}
Inversion of maps is one of the problems faced in computational and applied Sciences which has not received much attention especially in discrete mathematical domain. However this problem is closely related to Cryptanalysis and many other applications. Let $S$ be a finite set and $F:S\rightarrow S$ is a map specified in terms of mathematical representation of the elements of $S$ and the action of the map. Given a point $y$ in $S$ the problem of Local Inversion of $F$ at $y$ is to determine $x$ such that $y=F(x)$ or to show that a solution does not exist. In practical situations there are memory and time limitations to compute the inverse when it is known a-priori that a solution exists. Hence the important issue in computation is to determine an algorithm for computing a local inverse $x$, or prove no such solution exists and to determine the complexity of both of these computations to ascertian whether the computation is possible within the constraints of time and memory. A very concrete application of Local Inversion problem is that of Cryptanalysis of an encryption function $E(K,P)$ in which $K$ is the unknown key while $P$ is a known string of plaintext input and $C=E(K,P)$ is a known value of the function $E$. Hence the problem of recovering the key is the local inversion problem of solving $y=F(x)$ where $y=C$ and $F(x)=E(x,P)$.

\subsection{Previous work}
In previous papers \cite{sule322, sule521, sule722,suleCAM22} the author showed a solution of the local inversion problem for maps $F:\ff^n\rightarrow\ff^n$ using the concept of Linear Complexity (LC) of the sequence 
\beq\label{RecSeq}
S=\{s_k=F^{(k)}(y),k=0,1,2,\ldots\}
\eeq 
or iterative sequence $\{s_{(k+1)}=F(s_k),s_0=y\}$, when the sequence is periodic. $LC$ is the degree of the associated minimal polynomial which defines a recurrence relation of the elements of the sequence. Under the periodicity condition local inversion has a unique solution $x=F^{(N-1)}(y)$ where $N$ is the period. However in practical applications, this period $N$ is not known, nor is the sequence known beyond a limited part for $k=M-1$ where $M$ is of polynomial size in $n\log |\ff|$. The complexity of local inversion is then of polynomial order in the $LC$ of the given partial sequence, if it has the LC equal to the full periodic sequence.

In this paper the objective is to extend the above work to solve the local inversion problem using a generalization of the $LC$ for vector valued sequences to \emph{Word Linear Complexity} ($WLC$). In this generalization the recurrence relation satisfied by the recursive sequence $S$ of vector valued elements $s_k$ in $\ff^n$ is defined by an associated $n\times n$ matrix polynomial over $\ff$ instead of a scalar polynomial as in the case of $LC$. We show that the local inversion problem can also be solved using such matrix minimal polynomial when it exists and whose degree is a fraction of the scalar minimal polynomial by $n$. A condition for existence of a non trivial matrix minimal polynomial and the WLC is also derived.
\section{Matrix-Vector Recurrence Relations and Matrix minimal polynomial}
The local inversion problem of solving for $x$ such that $y=F(x)$ when $F$ is a map in a finite Cartesian space $\ff^n$ is ultimately the problem of computing the prefix of the sequence $S$ of iterates in (\ref{RecSeq}). The problem of predicting the next element $V_M$ of a finite sequence $\{V_0,V_1,\ldots,V_{(M-1)}\}$ is a well known classic problem that appears in literature in Number Theory and offered in puzzles. Newton's well known method of finite differences aims to solve a rule governing a given sequence to predict the next element. In fact when successive differences of sequence elements in Newton's method result into a constant sequence, both these problems are (that of predicting the next element or the prefix element) solved. The problem of predicting the prefix $V_{(-1)}$ of the given sequence has a unique solution for a iterative sequence $S$ only when $S$ is generated by the starting point $y$ such that the entire sequence is either periodic or is a chain which terminates at a periodic orbit. We shall not consider this later situation when $S$ is a chain. In practical situations the full periodic sequence is never available because the period is exponential. Hence local inversion has to be predicted from a limited sequence $S$ given upto $M$ terms. Hence the problem of predicting the prefix of a sequence is one of the poorly addressed question of computational mathematics. For the sake of completeness it is worthwhile to state a general problem of inversion of a sequence as follows.

\subsection*{Problem: Inversion of a Sequence}
Let $\ff$ be a finite field. Given a sequence of vectors $V(M)=\{V_0,V_1,\ldots,V_{(M-1)}\}$ where $V_i$ belong to $\ff^n$ for a given fixed $n$, determine existence and computation of a prefix vector $V_{(-1)}$ which is compatible with a rule of the sequence $V(M)$. (This problem is similar to predicting the next element $V_M$ of the sequence often encountered in puzzles). When the sequence is defined by a map $F:\ff^n\rightarrow\ff^n$ by iterations $V_{(i+1)}=F(V_i)$ for $i=0,1,2,\ldots,(M-1)$, then a local inverse $V_{(-1)}$ of $F$ at $V_0$ satisfies $F(V_{(-1)})=V_0$. Hence the problems is called as Local Inversion problem of $F$ at $V_0$.

\subsection{Recurrence relations satisfied by the sequence}
We now consider the problem of expressing the recurrence relation satisfied by a vector valued sequence $V(M)=\{V_i, i=0,1,2,\ldots,(M-1)\}$ of length $M$ where $V_i$ belong to $\ff^n$. This sequence is the first $M$ terms of the periodic sequence $V=\{V_0,V_1,\ldots,V_{(N-1)}\}$ of period $N\geq M$. The sequence $V(M)$ is said to satisfy a \emph{scalar recurrence relation} (RR) of order $m$ if there exist scalars $\al_i$ in $\ff$ such that the sequence $V(M)$ satisfies
\beq\label{ScalarRR}
V_{(m+j)}+\sum_{i=0}^{(m-1)}\al_iV_{(i+j)}=0
\eeq
for $j=0,1,2,\ldots,(M-m-1)$. We consider a generalization of such a scalar RR. The sequence $V(M)$ is said to satisfy a \emph{Matrix-vector recurrence relation} (MRR) of order $m$ if there exist matrices $A_i$, $i=0,1,2,\ldots,m$ in $\ff^{n\times n}$ such that
\beq\label{MatrixRR}
A_mV_{(m+j)}+\sum_{i=0}^{(m-1)}A_iV_i=0
\eeq
where the matrix co-efficients are operated from left on sequence elements $V_i$ represented as column vectors. 

However, the expression (\ref{MatrixRR}) also has the meaning that the sequence $V$ is determined by this relation for given matrices $A_i,i=0,1,2,\ldots,m$. Hence we assume that the recurrence relation with co-efficient matrices $A_i$ identifies the sequence $V$ uniquely given the initial sequence $V_0,V_1,\ldots,V_{(m-1)}$. Hence we impose the condition on the recurrence relation (\ref{MatrixRR}) that $V_m$ is determined uniquely by the previous $m$ vector elements $V_i$ of the sequence $V$. We shall always consider recurrence relations in the vector sequence with this uniqueness condition. Hence the matrix $A_m$ is always non-singular and the recurrence relation (\ref{MatrixRR}) is the same as
\beq\label{NormalisedMRR}
V_{(m+j)}+\sum_{i=0}^{(m-1)}A_iV_i=0
\eeq
where $A_i$ in this equation are now actually $(A_m)^{-1}A_i$ computed from the matrix co-efficients of the previous equation (\ref{MatrixRR}). We call (\ref{NormalisedMRR}) as \emph{normalised matrix recurrence relations} (NMRR). Analogous right acting MRR can be considered if $V_i$ are considered as row vectors. 

\subsection{Minimal matrix polynomial and the WLC}
Associated to the normalised MRR (\ref{NormalisedMRR}), is a matrix polynomial
\beq\label{MPoly}
P(X)=X^mI+\sum_{i=0}^{(m-1)}A_iX^i
\eeq
(the indeterminate $X$ is assumed to be commutative with the matrix co-effcients). The polynomial $P(X)$ is called an \emph{annihilating matrix polynomial} of the sequence $V(M)$. 

If the order $m$ of the normalised NMRR (\ref{NormalisedMRR}) is the smallest then $M(X)$ is called the \emph{minimal matrix polynomial} of $V(M)$. We denote the minimal matrix polynomial associated with the normalised recurrence relation of minimal order (\ref{NormalisedMRR}) as
\beq\label{MinimalPoly}
M(X)=X^m+\sum_{i=0}^{(m-1)}A_iV_i
\eeq
The degree of the minimal matrix polynomial is called the \emph{Word Linear Complexity} ($WLC$). The polynomial associated with the recurrence relation (\ref{ScalarRR}) is called just \emph{minimal polynomial} as it is a scalar minimal polynomial with co-efficients $\al_i$ and the degree of the scalar minimal polynomial is called $LC$. For the periodic sequence $V$, the polynomial $X^NI-I$ is a annihilating matrix polynomial where $I$ is the $n\times n$ identity matrix, $X^N-1$ is the scalar annihilating polynomial.

\subsection{Divisibility by the matrix minimal polynomial}
It is well established that the scalar minimal polynomial always divides a scalar annihilating polynomial. Does this also hold for matrix polynomials? We explore this question now. For any matrix polynomial $P(X)$ with $n\times n$ co-efficients over $\ff$ we have the Euclidean division theorem written in two ways, for left division and right division by another matrix polynomial $D(X)$.

\begin{theorem}[Euclidean division for matrix polynomials]\label{EuclidDiv}
Let $P(X)$ and $D(X)$ be $n\times n$ matrix polynomials over a field $\ff$ such that
\[
D(X)=D_mX^m+D_{(m-1)}X^{(m-1)}+\ldots+D_1X+D_0
\]
$\deg D(X)=m$, the matrix $D_{m}$ is non-singular and $\deg P(X)=N\geq m$. Then there exist unique matrix polynomials $Q_L(X)$, $Q_R(X)$ called the left and right quotients respectively and unique matrix polynomials $R_L(X)$ and $R_R(X)$ called left and right remainders respectively, all of size $n\times n$ such that
\begin{enumerate}
    \item
    $\left\{
    \begin{array}{lcl} 
    P(X) & = & D(X)Q_L(X)+R_L(X)\\
    P(X) & = & Q_R(X)D(X)+R_R(X)
    \end{array}\right.
    $
    \item $0\leq\deg R_L(X)<m$, $0\leq\deg R_R(X)<m$
\end{enumerate}
\end{theorem}
\begin{proof}
    Consider the first step of left division of $P(X)$ by $D(X)$. Let 
    \[
    P(X)=P_{N}X^N+P_{(N-1)}X^{(N-1)}+\ldots+P_1X+P_0
    \]
    where $P_N\neq 0$. By left division by $D(X)$ it is meant we want to find polynomials $Q(X)$, $R(X)$ such that 
    \[
    P(X)-D(X)Q(X)=R(X)
    \]
    Consider $Q(X)=Q_{(N-m)}X^{(N-m)}$. Then the highest degree co-efficient of the left hand side of the above equation is $P_N-D_mQ_{(N-m)}$. Since $D_m$ is non-singular, choosing the unique $Q_{(N-m)}=(D_m)^{-1}P_N$ it follows that
    \[
    \deg R(X)<N
    \]
    Repeating this process by algorithmic steps 

    \noindent
    \mbox{\textbf{Algorithm}: Division}
    \[
    \begin{array}{lcl}
    \mbox{Compute } R(X)\\
    \mbox{While } R(X) & \neq & 0 \mbox{ Repeate}\\
    P(X) & \leftarrow & R(X)\\
    Q_{(N-m)} & = & (D_m)^{-1}P_N\\
    R(X) & \leftarrow & P(X)-D(X)Q_{(N-m)}X^{(N-m)}
    \end{array}
    \]
    the degree of $R(X)$ continues to reduce at every step until $\deg P(X)<m$. Then the polynomial $Q(X)$ formed by co-efficient matrices of each step of this algorithm 
    \[
    Q(X)=Q_{(N-m)}X^{(N-m)}+Q_{(N-m-1)}X^{(N-m-1)}+\ldots+Q_1X+Q_0
    \]
    is unique and satisfies
    \[
    P(X)-D(X)Q(X)=R(X)
    \]
    where $\deg R(X)<m$.
    Similar constructive proof can be given for right division by $D(X)$ to prove the identities claimed.
\end{proof}

With the matrix polynomial divisibility when the highest degree co-efficient is non-singular, we have the following observation. Consider a vector sequence $V=\{V_i\}$ of period $N$ and let $M(X)$ denote the matrix minimal polynomial of $V$ as in (\ref{MinimalPoly}). Then we have

\begin{lemma}
Any matrix annihilating polynomial $P(X)$ of $V$ is left as well as right divisible by $M(X)$. 
\end{lemma} 
\begin{proof}
    Suppose $P(X)=P_pX^p+P_{(p-1)}X^{(p-1)}+\ldots+P_1X+P_0$ is an annihilating matrix polynomial of $V$ and let $M(X)=M_mX^m+\ldots+M_0$ be the minimal polynomial of degree $m$ then $p\geq m$. Note that the highest degree co-efficient matrix $M_m$ is nonsingular. Hence by the theorem (\ref{EuclidDiv}) by left division by $M(X)$ there exist unique matrix polynomials $Q(X)$, $R(X)$ such that the right division identity in the theorem holds.
    \[
    P(X)=Q_R(X)M(X)+R_R(X)
    \]
    where $\deg R_R(X)<m$. The polynomials $P(X)$ and $M(X)$ moreover are associated to recurrence relations
    \[
    P_pV_p+\sum_{i=0}^{(p-1)}P_iV_i=0,\mbox{   }M_mV_m+\sum_{i=0}^{(m-1)}M_iV_i=0
    \]
    Denoting $R(X)=R_rX^r-\sum_{i=0}^{(r-1)}$, the above recurrence relations lead to the recurrence relation
    \[
    R_rV_r-\sum_{i=0}^{(r-1)}R_iV_i=0
    \]
    However $m$ is the smallest degree of a recurrence relation which identifies the sequence $V$ uniquely. The recurrence relation associated with the polynomial $R(X)$ cannot identify the sequence $V$ uniquely because the order of this recurrence is less than the minimum order $m$ possible with the recurrence relation which identifies $V$ uniquely. Hence the only way all the three recurrence relations hold for $V_i$ is that $R_i=0$ for all $i$, for otherwise a sequence of components $v_{kj}$, $j=0,1,2,\ldots,(r-1)$ at some index $k$ of vectors $V_j$ may satisfy the recurrence defined by $R(X)$ but may not be consistent with recurrences defined by $P(X)$ and $M(X)$. Similar divisibility proof can be given for left division for the sequence $V_i$ considered as row vectors.
\end{proof}

If we choose the matrix minimal polynomial $M(X)$ with highest degree co-effcient matrix as $I$ then the polynomial is unique. This is because then for two distinct minimal polynomials of the same degree with highest degree term $X^mI$ there is a recurrence relation of the sequence of order $<m$ which also identifies the sequence uniquely. Since $m$ is the smallest order of such a relation, matrix polynomial $M(X)$ with highest degree term $X^mI$ is unique. Above lemma leads immediately to the following theorem.

\begin{theorem}
    Let the sequence $V$ be periodic of period $N$ and let there exist a nontrivial matrix minimal polynomial of $V$ denoted $M(X)$. Let the scalar minimal polynomial of $V$ be denoted $m(X)$. Then 
    \begin{enumerate}
        \item $M_0$ the constant co-efficient of $M(X)$ is non-singular.
        \item There exist distinct irreducible polynomials $p_i(X)$ such that 
        \[
        \det M(X)=\prod_i(p_i(X))^{e_i}
        \]
        for some exponents $e_i$ and $(\prod_ip_i(X))|(X^N-1)$. 
        \item $WLC\leq LC$
        \item $\det M(X)|(m(X))^n$ and $\deg\det M(X)\leq mn$.
        \item $\ord\det M(X)|Nn$
    \end{enumerate}
\end{theorem}

\begin{proof}
    When $V$ has period $N$, $X^NI-I$ is an annihilating polynomial of $V$ hence divisible by $M(X)$ by the above theorem. Let 
    \[
    X^NI-I=Q(X)M(X)
    \]
    then $Q_0M_0=I$ which proves that $M_0$ is non-singular.

    Next let $\zeta$ be a root of $\det M(X)$. Then there exists $v\neq 0$ in $K^n$ such that $M(\zeta)v=0$ where $K$ is the algebraic closure of $\ff$. Then
    \[
    Q(\zeta)M(\zeta)v=(\zeta^N-1)v=0
    \]
    Hence $\zeta$ is a root of $(X^N-1)$. Hence all roots of $\det M(X)$ are contained in roots of $X^N-1$. Hence all distinct irreducible factors of $\det M(X)$ divide $(X^N-1)$.

    If the scalar minimal polynomial $m(X)$ of $V$ has degree $m$, it gives the recurrence relation
    \[
    V_{m+j}=\sum_{i=0}^{(m-1)}\al_iV_i
    \]
    Hence the associated matrix polynomial $m(X)I$ is an annihilating polynomial of $V$. Hence $M(X)$ divides $m(X)I$ by both left and right division. Consequently $\deg M(X)\leq\deg m(X)$. Hence we have the inequality
    \[
    WLC\leq LC
    \]
    Let $Q(X)M(X)=m(X)I$. If $\zeta$ is a root of $\det M(X)$ then there is a vector $v\neq 0$ in $K^n$ such that $M(\zeta)v=0$ where $K$ is the algebraic closure of $\ff$. Hence $Q(\zeta)M(\zeta)v=m(\zeta)v=0$. Since $v\neq 0$, this implies $m(\zeta)=0$. This shows that all roots of $\det M(X)$ are contained in roots of $m(X)$. Hence $\det M(X)|(m(X))^n$ and also that $\deg\det M(X)\leq mn$.

    Since order of $m(X)=N$ and $\det M(X)|(m(X))^n$, 
    \[
    \ord\det M(X)|\ord (m(X))^n
    \]
\end{proof}

\section{Computation of the matrix minimal polynomial and local inverse}
For a periodic sequence $V$ of period $N$ the scalar minimal polynomial $m(X)$ always exists. The degree $m$ of $m(X)$ which is called the $LC$ of $V$ is the number $m$ such the sequence of Hankel matrices of increasing size has its maximal rank $m$ equal to the number of columns of the Hankel matrix. Solution to the minimal polynomial is the unique solution of the recurrence relations of order $m$. Consider the Hankel matrix defined by a subsequence $V(k)=\{V_0,V_1,\ldots,V_{(2k-2)}\}$ of length $(2k-1)$ denoted as 
\beq\label{Hankelk}
H(k)=
\left[\begin{array}{llll}
V_0 & V_1 & \ldots & V_{(k-1)}\\
V_1 & V_2 & \ldots & V_{k}\\
\vdots & \vdots & \ldots & \vdots\\
V_{(k-1)} & V_{k} & \ldots & V_{(2k-2)}
\end{array}\right]
\eeq
Then the $LC$ of the sequence $V$ of period $N\geq 2k$, is $m$ when 
\beq\label{Rankcond}
m=\rank H(m)=\rank H(m+j)
\eeq
for $j=1,2,\ldots, \lfloor(N-2m)\rfloor$. The co-efficients $\bar{\alpha}^T=(\al_0,\al_1,\ldots,\al_{(m-1)})$ of the minimal polynomial are computed as the unique solution of the linear system of equations
\beq\label{Eqnforminpoly}
H(m)\bar{\alpha}=-h(m+1)
\eeq
where $h_{(m+1)}$ is the last column of $H(m+1)$ after dropping the bottom-most entry $V_{(2m)}$,
\[
h(m+1)^T=(V_m,V_{(m+1)},\ldots,V_{(2m-1)})
\]
The equations (\ref{Eqnforminpoly}) are the recurrence relations of order $m$ satisfied by the sequence $V$. For $LC=m>\lfloor N/2\rfloor$, periodically repeating sequence elements of $V$ are entered to complete the Hankel matrix. In the worst case when $m=N$ the minimal polynomial is $X^N-1$ itself. This completes the description of the existence and computation of the scalar minimal polynomial.

\subsection{Condition for existence of a minimal matrix polynomial}
We now describe the analogous rank condition of (\ref{Rankcond}) for existence of a minimal matrix polynomial $M(X)$ and the system of equations analogous to (\ref{Eqnforminpoly}) arising from the recurrence relations to solve the matrix co-efficients of $M(X)$. For simplicity of expressions of the matrix recurrence relations we shall consider the normalised relations (\ref{NormalisedMRR}) since the relation is required to identify the sequence $V(M)$ uniquely. Let $m$ be the order of the minimal order recurrence relation (\ref{NormalisedMRR}) then the co-efficients of matrix polynomial $M(X)$ in (\ref{MinimalPoly}) satisfy the recurrence relations
\beq\label{MinpolyRR}
V_{(m+j)}=\sum_{i=0}^{(m-1)}A_iV_i
\eeq
for $j=0,1,2,\ldots, (M-m-1)$. Hence the condition for existence of a matrix polynomial $M(X)$ and its unique solvability is given by the above linear equations (\ref{MinpolyRR}). Consider the Hankel matrix $\tilde{H}(nm)$ as defined by the vector sequence $V(M)$ in $\ff^n$ for $n>1$ as
\beq\label{VHankelmatrix}
\tilde{H}(nm)=
\left[\begin{array}{llll}
V_0 & V_1 & \ldots & V_{(nm-1)}\\
V_1 & V_2 & \ldots & V_{nm}\\
\vdots & \vdots & \ldots & \vdots\\
V_{(m-1)} & V_{k} & \ldots & V_{((n+1)m-2)}
\end{array}\right]
\eeq

\begin{theorem}
Given a sequence $V(M)$, where $M$ is even, there exists a unique matrix minimal polynomial $M(X)$ of degree $d$ iff
\beq\label{RankcondmatrixRR}
dn=\rank \tilde{H}(dn)=\rank \tilde{H}((d+j)n)
\eeq
for $j=1,2,3,\ldots,\frac{M/2-dn}{n}$. The matrix co-efficients $A_i$ are the unique solution of the equation
\beq\label{Eqnformatrixminpoly}
\left[A_0,A_1,\ldots,A_{(d-1)}\right]\tilde{H}(dn)=-
\left[V_{dn},V_{(dn+1},\ldots,V_{(2dn-1)}\right]
\eeq
\end{theorem}

The rank condition includes the case of the scalar minimal polynomial when the degree $d=m$ where $m$ is the degree of the scalar minimal polynomial. When this happens all the coefficient matrices $A_i$ are scalar matrices of size $n$. Note that the matrix $\tilde{H}(dn)$ is not the same as the Hankel matrix $H(k)$ of (\ref{Hankelk}). The co-efficients of the matrix minimal polynomial operate on the the matrix $\tilde{H}(.)$ on left. When $d=m$ the scalar nature of the matrix co-efficients allows commuting the co-effcients to the right and make them equivalent to the equations (\ref{Eqnforminpoly}) for solving the scalar minimal polynomial.

\begin{proof}
    If the matrix minimal polynomial $M(X)$ of degree $d$ exists then the matrix recurrence relations (\ref{MinpolyRR}) are satisfied by $V(M)$. Since $M(X)$ is unique, the rank conditions (\ref{RankcondmatrixRR}) hold for the sequence $V_i$ for $i=0,1,2,\dots, M-1$. Conversely, given these rank conditions (\ref{RankcondmatrixRR}), the solution of the matrix coefficients $A_i$ of $M(X)$ are unique solutions of the equation (\ref{Eqnformatrixminpoly}) and satisfy the recurrence relations (\ref{MinpolyRR}).
\end{proof}

Note that even if the period of a sequence is $M$ and is odd, we can consider an even length sequence by adding an extra element (the first repeating element) to the sequence to make the length even. Hence the condition of the theorem insisting $M$ is even is not a serious restriction. What is not always satisfied however, is the rank condition for any $d$ for a given sequence such that $dn\leq M/2$. In that case the only possibility is that the minimal polynomial is not a nontrivial matrix polynomial, but equals the scalar minimal polynomial and the $\mbox{WLC}=\mbox{LC}$. Next theorem clarifies the relationship between the LC and the WLC of a sequence $V(M)$.

\begin{theorem}
    Given a sequence $V(M)$ in $\ff^n$ where $M$ is even and $m$ the $LC$ of $V(M)$. If there exists a matrix minimal polynomial $M(X)$ of degree $d<m$ then $nd=m$. Hence
    \[
    LC=n(WLC)
    \]
\end{theorem}

\begin{proof}
    \emph{Necessity}\/: Given the LC $m$ and degree of $M(X)$ to be $d$, $m=\rank H(m)=\rank H(m+j)$ for $j=1,2,\ldots,(M/2-1)$. Moreover from the previous theorem $\rank H(dn)=m=\rank H((d+j)n)$ for $j=0,1,2,\ldots,\frac{M/2-dn}{n}$. Hence $dn=m$.

    \emph{Suffciency}\/: If $m=dn$ then the rank conditions (\ref{RankcondmatrixRR}) are satisfied. Hence there is a unique matrix polynomial $M(X)$ of degree $d$ which satisfies the recurrence relations (\ref{MinpolyRR}). For any degree $\tilde{d}>d$ the matrix $H(n\tilde{d})$ cannot be full rank equal to $nd$ since $m=nd$ is the LC. 
\end{proof}

\subsection{Local inverse of a map and inverse of a sequence}
In this subsection we develop the formula for the local inverse $x$ of the map equation $y=F(x)$ for a map $F:\ff^n\rightarrow\ff^n$ using the matrix minimal polynomial of the iterated sequence $V_{i+1}=F(V_i)$ when $y=V_0$. As described in the introduction, this inverse also gives the prefix element $V_{(-1)}$ of the sequence $V(M)=\{V_0,V_1,\ldots,V_{(M-1)}\}$ even when it is not generated iterativly by a map $F$, when the sequence is periodic and the minimal polynomial of the full periodic sequence is the same as that of $V(M)$. 

The local inverse of $y=F(x)$ when the matrix minimal polynomial is known for the periodic and partially given iterated sequence $V(M)=\{V_{i+1}=F(V_i),V_0=y\}$ is given by the following theorem.

\begin{theorem}
    Let $V(M)$ be a given subsequence of a periodic sequence $V$ in $\ff^n$ of period $N$ and has the matrix minimal polynomial $M(X)$ defined in (\ref{MinimalPoly}). Then there is a unique solution $V_{(-1)}$ in $\ff^n$ given by
    \[
    V_{(-1)}=-A_0^{-1}[V_{(m-2)}-\sum_{i=0}^{i=(m-1)}A_iV_i]
    \]
    such that the recurrence relations (\ref{MatrixRR}) of $V(M)$ are also satisfied by 
    \[
    \{V_{(-1)},V_0,V_1,\ldots,V_{(M-2)}\}.
    \]
    If $V_0=F(x)$ then $V_{(-1)}$ is the unique local inverse of $F$ at $V_0$.
\end{theorem}

\begin{proof}
    Since the sequence $V(M)$ is a subsequence of a periodic sequence $V$ the shifted sequence $\{V_{(N-1)},V_0,V_1,\ldots,V_{(N-2)}\}$ has the same minimal polynomial $M(X)$. Hence $V_{(-1)}=V_{(N-1)}$ exists uniquely which confirms with the matrix recurrence relations (\ref{MatrixRR}) of $V(M)$. Hence
    \[
    V_{(m-2)}+\sum_{i=-1}^{i=(m-1)}A_{(i+1)}V_i=0
    \]
    Since the co-efficient matrix $A_0$ is non-singular it follows that
    \[
    A_0^{-1}[V_{(m-2)}-\sum_{i=0}^{i=(m-1)}A_iV_i]=-V_{(-1)}
    \]
    If the sequence is an iterative sequence generated by $F$ starting from $V_0$, then $V_0=F(V_{(-1)}$ hence this is the unique solution of the local inverse. 
\end{proof}

\section{Examples}
In this section we gather examples to illustrate the existence, computation and the role of the matrix minimal polynomial of vector valued sequences.

\begin{example}
    Consider the vector sequence arranged as columns of the array
    \[
    \left[
    \begin{array}{llllllll}
    1 & 0 & 0 & 0 & 1 & 0 & 1 & 1\\
    0 & 0 & 0 & 1 & 0 & 1 & 1 & 0\\
    0 & 0 & 1 & 0 & 1 & 1 & 1 & 0
    \end{array}
    \right]
    \]
    The sequence is repeated after from the first vector. Hence the period of the vector sequence is $7$. The Hankel matrix $H(6)$ for the first sequence reaches its maximal rank $6$. The minimal polynomial of the first sequence is computed by solving the co-effcients from the equation (\ref{Eqnforminpoly}) as
    \[
    X^6+X^5+X^4+X^3+X^2+X+1
    \]
    The second sequence reaches the maximal rank of $H(7)=7$ which is the period. Hence the minimal polynomial of the second sequence is $X^7+1$ Hence taking lcm the minimal polynomial of the vector sequence is $m(X)=X^7+1$. Since vector length is $n=3$, $n$ does not divide the degree of the scalar minimal polynomial of the sequence. Hence for this sequence there is no nontrivial matrix minimal polynomial but
    \[
    M(X)=m(X)I
    \]
    i.e. matrix minimal polynomial is the same as the scalar minimal polynomial.
\end{example}

\begin{example}
    Consider the periodic vector sequence of period $6$.
    \[
    \left[
    \begin{array}{llllll}
    1 & 0 & 0 & 0 & 1 & 0\\
    0 & 0 & 1 & 1 & 0 & 0
    \end{array}
    \right]
    \]
    The second sequence has maximal rank Hankel matrix $H(6)$ hence the minimal polynomial of the vector sequence is $X^6+1$ and the LC is $m=6$. The vector length $n=2$ divides LC $d=m/n=3$. However the Hankel matrix $\tilde{H}(2*3)$ of the vector sequence is given by
    \[
    \barr{llllll}
    1 & 0 & 0 & 0 & 1 & 0\\
    0 & 0 & 1 & 1 & 0 & 0\\
    \hline
    0 & 0 & 0 & 1 & 0 & 1\\
    0 & 1 & 1 & 0 & 0 & 0\\
    \hline
    0 & 0 & 1 & 0 & 1 & 0\\
    1 & 1 & 0 & 0 & 0 & 0
    \earr
    \]
    has rank $5$. Hence $\tilde{H}(2*3)$ is not non-singular. Hence no unique nontrivial matrix minimal polynomial exists for this vector sequence. Alternatively $M(X)=m(X)I=(X^6+1)I$ is a scalar minimal polynomial.
\end{example}

\begin{example}
Consider the vector sequence in $\ftwo^2$ given by
\[
\left[
\begin{array}{llllll}
1 & 0 & 0 & 0 & 1 & 0  \\
1 & 0 & 1 & 1 & 0 & 0 
\end{array}
\right]
\]
Minimal polynomials of individual component sequences are $X^6+1$ and $X^4+X^2+1$. Hence the scalar minimal polynomial of the vector sequence is $X^6+1$. It is discovered that $\tilde{H}(6)=\tilde{H}(2*3)$ has the maximal rank $6$ hence the degree of the matrix minimal polynomial is $d=3$. Let the matrix minimal polynomial be expressed as
\[
M(X)=X^3-(A_2X^2+A_1X+A_0)
\]
then the $2\times 2$ co-efficient matrices $A_i$ satisfy the equation (\ref{Eqnformatrixminpoly}) with $d=3$ and $n=2$ over $\ftwo$. Solving this equation gives
\[
M(X)=X^3I+A_2X^2+A_1X+A_0
\]
where
\[
\left[A_0,A_1,A_2\right]=
\barr{ll|ll|ll}
0 & 1 & 0 & 0 & 1 & 1\\
1 & 1 & 0 & 1 & 0 & 1
\earr
\]
The local inverse $V_{(-1)}$ of $F(V_{(-1)})=V_0$ is obtained as
\[
V_{(-1)}=(A_0)^{-1}(V_3-A_2V_1-A_1V_0)=
\barr{l}0\\0\earr
\]
which is correct as the sequence repeats.

Finally the divisibility of $(X^6+1)I$ by $M(X)$ is established by dividing as in the algorithm in Theorem 1, the matrix polynomial $Q(X)$ such that $(X^6+1)I=M(X)Q(X)$ is given by $Q(X)=X^3+Q_2X^2+Q_1X+Q_0$ where
\[
\left[Q_2,Q_1,Q_0\right]=
\barr{ll|ll|ll}
1 & 1 & 1 & 0 & 1 & 1\\
0 & 1 & 0 & 0 & 1 & 0
\earr
\]
The polynomial $Q(X)$ is computed by pure division and does show divisibility which is proved using the recurrence relations as in the proof of Theorem 1.
\end{example}
\section{Conclusion}
Local inversion of a map at a point in its image is synonymous with finding a prefix element of the sequence of iterates of the map starting from the point. For periodic sequences there always exists a minimal polynomial whose degree is called the $LC$. The local inverse of the map (or the prefix element of the sequence) is unique and can be solved using the linear recurrence relations satisfied by the sequence as defined by the minimal polynomial. The notion of $LC$ and the minimal polynomial can be extended to $WLC$ and the matrix minimal polynomial. However a nontrivial matrix minimal polynomial does not always exist for a given vector valued sequence and when it exists $LC=n(WLC)$ where $n$ is the number of components of the vector sequence. If a non trivial matrix minimal polynomial does not exist then the scalar minimal polynomial of degree $LC$ is the only minimal polynomial. The condition for existence of a nontrivial matrix minimal polynomial arises because the recurrence relations defined by the matrix minimal polynomial exhibit presence of mixing of the components of vector sequences unlike the recurrence relations defined by the scalar minimal polynomial which are restricted within the sequence of each component independently.

\begin{center}
    Acknowledgement

    \noindent
    Author is thankful to A.\ Ramachandran for useful discussions and help in corrections during writing of this paper.
\end{center}

\end{document}